\documentclass{scrartcl}

\usepackage[british]{babel}

\usepackage{amsmath,amssymb,amsthm}
\theoremstyle{plain}
\newtheorem{theorem}{Theorem}
\newtheorem{lemma}[theorem]{Lemma}
\newtheorem{cor}[theorem]{Corollary}
\theoremstyle{definition}
\newtheorem*{remark}{Remark}
\theoremstyle{remark}

\usepackage[perpage,symbol*]{footmisc}

\let\deph\emph
\DeclareMathOperator{\sgn}{sgn}
\def\R{\mathbb{R}}
\def\Q{\mathbb{Q}}
\def\tr{^{\mathsf T}}
\def\Y{\{-1,1\}^{n}}
\def\Deltayz{\Delta_{y,z}}
\def\Deltatz{\Delta_{t,z}}

\usepackage{bbm}
\def\1{\mathbbm{1}}

\begin{document}

\title{P-matrix recognition is co-NP-complete}
\author{Jan Foniok%
\\
\small ETH Zurich, Institute for Operations Research\\[-5pt]
\small R\"amistrasse 101, 8092 Zurich, Switzerland\\[-5pt]
\small\texttt{foniok@math.ethz.ch}%
}
\date{18 October 2007}

\maketitle

This is a summary of the proof by G.E.~Coxson~\cite{Cox:The-P-matrix}
that P-matrix recognition is co-NP-complete. The result follows by
a reduction from the MAX CUT problem using results of S.~Poljak and
J.~Rohn~\cite{PolRoh:Checking-robust}.

\section{Considered problems}

Our main interest is the complexity of deciding whether an input matrix
is a P-matrix. A \deph{P-matrix} is a square matrix $M\in\R^{n\times
n}$ such that all its principal minors are positive. Such matrices were
first studied by Fiedler and Pt\'ak~\cite{FiePta:On-matrices-with}.

\subsection*{P-MATRIX}
\begin{tabular}[t]{lp{11cm}}
\textbf{Instance:}&
A square matrix $M\in\Q^{n\times n}$.\\
\textbf{Question:}&
Are all the principal minors of~$M$ positive?
\end{tabular}

\bigskip

To start with, we use a well-known combinatorial problem.

\subsection*{SIMPLE MAX CUT}
\begin{tabular}[t]{lp{11cm}}
\textbf{Instance:}&
A graph $G=(V,E)$, a positive integer~$K$.\\
\textbf{Question:}&
Is there a partition of the vertex set~$V$ into sets $V_1$ and~$V_2$
such that the number of edges with one end in~$V_1$ and the other end
in~$V_2$ is at least~$K$?
\end{tabular}

\bigskip

Garey, Johnson and Stockmeyer~\cite{GarJohSto:Some-simplified} showed
that the SIMPLE MAX CUT problem is NP-complete.

\bigskip

The reduction from SIMPLE MAX CUT to P-MATRIX uses two intermediate
steps. The first of them is the computation of the $r$-norm of a matrix.

For an arbitrary matrix $A\in\R^{n\times n}$, let
\[r(A) = \max\left\{z\tr\! A y : z,y\in\Y\right\}. \]

\begin{remark}
The function $r$ is a matrix norm.
\end{remark}

\begin{proof}
For an arbitrary square matrix~$A$, we have $r(A)\ge0$ because $z\tr\!Ay = - (-z)\tr\!Ay$.
Moreover if $r(A)=0$, then $z\tr\!Ay=0$ for all choices of $z,y\in\Y$, hence $A=0$. If
$k\in\R$, then $z\tr(kA)y=k\cdot z\tr\!Ay$, so $r(kA)=|k|\cdot r(A)$.

Let $A,B\in\R^{n\times n}$. Then
\begin{multline*}
r(A+B) =\max\{z\tr(A+B)y: y,z\in\Y\}
=\max\{z\tr\!Ay + z\tr By: y,z\in\Y\}\\
\leq \max\{z\tr\!Ay : y,z\in\Y\} + \max\{z\tr By:y,z\in\Y\}\\
= r(A)+r(B).
\end{multline*}
Thus $r$ is also subadditive.
\end{proof}

The decision problem corresponding to $r$-norm computation is defined as follows.

\subsection*{MATRIX R-NORM}
\begin{tabular}[t]{lp{11cm}}
\textbf{Instance:}&
A matrix $A\in\Q^{n\times n}$ and a rational number~$K$.\\
\textbf{Question:}&
Is $r(A)\geq K$?
\end{tabular}

\bigskip

For the last of the decision problems considered here, we need the
notion of matrix interval.  If $A_-$ and $A_+$ are $n\times n$ real
matrices such that $A_- \leq A_+$ (that is, for each $r$ and~$s$ we
have $(A_-)_{r,s} \leq (A_+)_{r,s}$), then the \deph{matrix interval}%
	\footnote{This object is usually called an \emph{interval
	matrix}. Since it is actually an \emph{interval} and not a
	\emph{matrix}, I beg the reader to pardon my decision to call
	it an uncommon but appropriate name.}
$[A_-,A_+]$ is the set of all matrices~$A$ satisfying $A_-\le A \le A_+$.

A matrix interval is \deph{singular} if it contains a singular matrix;
otherwise it is \deph{non-singular}.

The decision problem we consider consists in testing whether a given
matrix interval is singular. We will see that this is a computationally
hard problem even when the difference $A_+-A_-$ has rank~$1$.

\subsection*{RK1-MATRIX-INTERVAL SINGULARITY}
\begin{tabular}[t]{lp{11cm}}
\textbf{Instance:}&
A non-singular matrix $A\in\Q^{n\times n}$ and a non-negative
matrix~$\Delta\in\Q^{n\times n}$ of rank~$1$.\\
\textbf{Question:}&
Is the matrix interval $[A-\Delta, A+\Delta]$ singular?
\end{tabular}

\bigskip

The rest of this exposition contains three polynomial reductions of these
problems, ultimately proving that P-MATRIX is co-NP-complete.

\section{Reduction from SIMPLE MAX CUT to MATRIX R-NORM}

Let $G=(V,E)$ be an undirected graph with $n=|V|$ and let $\ell=2|E|+1$. If $A(G)$ is the
adjacency matrix of~$G$, define $A=\ell\cdot I_n - A(G)$. Thus
\[ A_{u,v} =
\begin{cases}
\ell&	\text{if $u=v$,}\\
-1&	\text{if $uv\in E$,}\\
0&	\text{otherwise.}
\end{cases}
\]

Observe that for $y,z\in\Y$ we have $z\tr\!Ay \le y\tr\!Ay$ because of
the choice of~$\ell$.  Hence $r(A)=y\tr\!Ay$ for some $y\in\Y$.

Let $S\subseteq V$~be defined by $S=\{u: y_u=1\}$ and let $m'$~be
the number of edges of~$G$ with one end in~$S$ and the other end
in~$V\setminus S$. In this way, $m'$~is the size of the cut defined by
$S$ and~$V\setminus S$.

Then
\[y\tr\!Ay = n\ell + 4m' - 2|E|\]
and therefore there is a cut in~$G$ of size at least~$K$ if and only if
$r(A)\ge n\ell - 2|E| + 4K$.

The described reduction (by Poljak and Rohn~\cite{PolRoh:Checking-robust})
establishes the hardness of computing the $r$-norm.

\begin{theorem}
MATRIX R-NORM is NP-complete, even if input is restricted to non-singular matrices.
\end{theorem}

\begin{proof}
It follows from the reduction above that MATRIX R-NORM is NP-hard. Observe
that by the choice of~$\ell$ the matrix~$A$ in the reduction is strictly
diagonally dominant and thus non-singular.

A non-deterministic Turing machine can guess the values of $y,z\in\Y$
and check in polynomial time that $z\tr\!Ay\ge K$, so the problem is in
the class~NP.
\end{proof}

\section{Reduction from MATRIX R-NORM to RK1-MATRIX-INTERVAL SINGULARITY}

For a matrix $A\in\R^{n\times n}$ define
\[\rho_0(A) = \max\{|\lambda|: \text{$\lambda$ is a real eigenvalue of $A$}\}\]
and set $\rho_0(A)=0$ if $A$~has no real eigenvalue.

Further for a vector $y\in\R^n$ define $D(y)$ to be the diagonal $n\times
n$ matrix with diagonal vector~$y$.

The following fact was proved by Rohn~\cite{Roh:SysLinInt}.

\begin{lemma}
\label{lem:2}
Let $A$ be a real non-singular $n\times n$ matrix and let $\Delta$ be a
real non-negative $n\times n$ matrix. Then the matrix interval $[A-\Delta,
A+\Delta]$ is singular if and only if
$\rho_0(A^{-1}D(y)\Delta D(z)) \ge 1$ for some $y,z\in\Y$.
\end{lemma}

\begin{proof}
For $y,z\in\Y$ let $\Deltayz$ denote the matrix $D(y)\Delta D(z)$.

First suppose that $A^{-1}\Deltayz$ has a real eigenvalue~$\lambda$
such that $|\lambda|\ge1$ and $A^{-1}\Deltayz x = \lambda x$ for some
$y,z\in\Y$ and a non-zero vector~$x$. Then
\begin{gather*}
\left(1-\tfrac{1}{\lambda}A^{-1}\Deltayz\right) x = 0,\\
\left(A-\tfrac{1}{\lambda}\Deltayz\right) x = 0.
\end{gather*}
Hence $A-(1/\lambda)\Deltayz$ is a singular matrix in the interval $[A-\Delta, A+\Delta]$
because
\[\left|\tfrac{1}{\lambda}\Deltayz\right|=
\left|\tfrac{1}{\lambda}D(y)\Delta D(z)\right| \leq \Delta. \]
Therefore the interval $[A-\Delta, A+\Delta]$ is singular.

To prove the converse, suppose that $B$ is a singular matrix, $B\in[A-\Delta,A+\Delta]$. Let
$x$ be a non-zero vector for which $Bx=0$.

For $i=1,2,\dotsc,n$ set
\[ t_i = \frac{(Ax)_i}{(\Delta|x|)_i}. \]
We claim that $t\in[0,1]^n$. Indeed, $|Ax|=|(A-B)x|\le \Delta|x|$
because $Bx=0$ and $B\in[A-\Delta,A+\Delta]$.

Moreover, set $z=\sgn x$. Then $D(z)x=|x|$ and
\[(A-\Deltatz)x = Ax- D(t)\Delta D(z) x =
Ax- D(t)\Delta|x| = 0\] by the definition of~$t$.
Thus the matrix $A-\Deltatz$ is a singular matrix in the interval
$[A-\Delta, A+\Delta]$.

Define $\psi(s)=\det(A-\Delta_{s,z})$. The function~$\psi$
is affine in each of the variables $s_1,\dotsc,s_n$. Since
$\psi(t)=\det(A-\Deltatz)=0$, either there exists $y\in\Y$ such that
$\det(A-\Deltayz)=0$, or there exist $y,y'\in\Y$ such that
$\det(A-\Deltayz)\cdot\det(A-\Delta_{y',z})<0$.

In the latter case, without loss of generality we may assume that
$\det A\cdot\det(A-\Deltayz)<0$. The function~$\phi$ defined
by $\phi(\alpha)=\det(A-\alpha\Deltayz)$ is continuous and
$\phi(0)\phi(1)<0$, so $\phi$~has a root in~$(0,1)$.

In either case, there exist $y\in\Y$ and $\alpha\in(0,1]$ such that
$\det(A-\alpha\Deltayz)=0$. Then
\begin{gather*}
\det\left(\tfrac{1}{\alpha}A-\Deltayz\right) = 0,\\
\det\left(\tfrac{1}{\alpha}I-A^{-1}\Deltayz\right) = 0,
\end{gather*}
hence $\frac{1}{\alpha}$ is a real eigenvalue of the matrix
$A^{-1}D(y)\Delta D(z)$ and $\frac{1}{\alpha}\ge1$, as we were
supposed to prove.
\end{proof}

This lemma provides a useful connection between singularity of matrix
intervals and a parameter $\rho_0$ dependent on the two matrices
$A$,~$\Delta$ that define the interval. Next we establish a connection
between $\rho_0$ and the $r$-norm of matrices.

From now on let $\1$ be the all-one vector~$(1,1,\dotsc,1)\in\R^n$
and let $J=\1\cdot\1\tr$ be the all-one $n\times n$ matrix.

\begin{lemma}
\label{lem:3}
Let $A\in\R^{n\times n}$ be a non-singular matrix, let $\alpha$ be a positive
real number and let $\Delta=\alpha J$. Then
\[\max\left\{\rho_0(A D(y)\Delta D(z)) : y,z\in\Y\right\}
	= \alpha\cdot r(A).\]
\end{lemma}

\begin{proof}
First observe that $D(y)\Delta D(z)=\alpha\cdot D(y)\1\cdot\1\tr D(z)=
\alpha\cdot yz\tr$ for arbitrary $y,z\in\Y$.  If $\lambda$~is a non-zero
real eigenvalue of~$\alpha\cdot A yz\tr$ and $x$~is a non-zero vector
such that
\begin{align*}
\alpha\cdot A yz\tr x &= \lambda x \ne 0,\\
\intertext{then $z\tr x \ne 0$ and}
\alpha\cdot z\tr\!A yz\tr x &= \lambda\cdot z\tr x,\\
\alpha\cdot z\tr\!A y &= \lambda.
\end{align*}
Thus $\rho_0(A D(y)\Delta D(z))=\alpha\cdot|z\tr\!A y|$.
Hence
\begin{multline*}
\max\left\{\rho_0(A D(y)\Delta D(z)) : y,z\in\Y\right\}\\
	=\alpha\cdot\max\left\{|z\tr\!A y| : y,z\in\Y\right\}
	=\alpha\cdot r(A).
\end{multline*}
\end{proof}

Now everything is set for Poljak and Rohn's reduction~\cite{PolRoh:Checking-robust}.

\begin{theorem}
Let $A\in\R^{n\times n}$ be a non-singular matrix, let $K$~be a positive
real number and let $\Delta=(1/K)\cdot J$. Then $r(A)\ge K$ if and only if
the matrix interval $[A^{-1}-\Delta, A^{-1}+\Delta]$ is singular.
\end{theorem}

\begin{proof}
By Lemma~\ref{lem:2}, the matrix interval $[A^{-1}-\Delta,
A^{-1}+\Delta]$ is singular if and only if
$\rho_0(AD(y)\Delta D(z)) \ge 1$ for some $y,z\in\Y$.
By Lemma~\ref{lem:3},
$\rho_0(AD(y)\Delta D(z)) \ge 1$ for some $y,z\in\Y$
if and only if $r(A) \ge K$.
\end{proof}

\begin{cor}
\label{cor:sing}
RK1-MATRIX-INTERVAL SINGULARITY is NP-hard.
\qed
\end{cor}

\begin{remark}
Poljak and Rohn~\cite{PolRoh:Checking-robust} show that
RK1-MATRIX-INTERVAL SINGULARITY belongs to the class NP by proving the
existence of a singular matrix in every singular matrix interval, with
a polynomial bound on the size of all entries of that matrix.
\end{remark}

\section{Reduction from RK1-MATRIX-INTERVAL SINGULARITY to P-MATRIX}

The described reduction is by Coxson~\cite{Cox:The-P-matrix}.

Let $A,\Delta\in\R^{n\times n}$.
Consider the matrix interval $[A, A+\Delta]$. Let $\Delta^{i,j}$ be the
matrix whose element in the $i$th row and $j$th column is~$\Delta_{i,j}$
and which has zeros elsewhere. Then each matrix~$M$ in the interval
$[A,A+\Delta]$ can be uniquely expressed as
\begin{equation}
\label{eq:1}
M = A + \sum_{i,j=1}^{n} p_{i,j} \Delta^{i,j},
\end{equation}
where $p_{i,j}\in[0,1]$ for all values of~$i,j$.

Each matrix $\Delta^{i,j}$ is a rank-1 matrix (even if $\Delta$~has
higher rank), and so $\Delta^{i,j}=r_{i,j}s_{i,j}\tr$ for some
vectors $r_{i,j},s_{i,j}\in\R^n$. We can actually take $r_{i,j}$ to
be~$\Delta_{i,j}$ in its $i$th entry and zero elsewhere, and $s_{i,j}$
to be~$1$ in its $j$th entry and zero elsewhere.

Now let $R$ be the matrix whose columns are all the $n^2$
vectors~$r_{i,j}$ and let $S$ be the matrix whose columns are all the
$n^2$ vectors~$s_{i,j}$. Thus $\Delta=RS\tr$. Moreover, if $p\in\R^{n^2}$
is the vector formed by the numbers~$p_{i,j}$, we can write~\eqref{eq:1} as
\[M = A + R D(p) S\tr.\]

Suppose that $A$~is non-singular. Then the matrix interval $[A,A+\Delta]$
is non-singular if and only if
\begin{equation}
\label{eq:2}
\det(A+RD(p)S\tr) = \det(A)\det(I_n + A^{-1}RD(p)S\tr)\ne 0
\end{equation}
for each vector $p\in[0,1]^{n^2}$.

Supposing that the matrix~$A$ is non-singular, inequality~\eqref{eq:2}
holds if and only if
\begin{equation}
\label{eq:3}
\det(I_n + A^{-1}RD(p)S\tr) \ne 0.
\end{equation}

In this way we have proved that for a non-singular matrix~$A$,
singularity of the matrix interval $[A,A+\Delta]$ is equivalent to
the existence of a vector $p\in[0,1]^{n^2}$ that does not satisfy
inequality~\eqref{eq:3}. Since the expression in~\eqref{eq:3} is a
multi-affine function of~$p$, we can actually derive another condition.

\begin{lemma}
\label{lem:multaf}
Let $\psi(p)=\det(I_n+A^{-1}RD(p)S\tr)$. Then inequality~\eqref{eq:3}
holds for each $p\in[0,1]^{n^2}$ if and only if $\psi(p)>0$ for each
$p\in\{0,1\}^{n^2}$.
\end{lemma}

\begin{proof}
First observe that $\psi(p)=\det(I_n+A^{-1}RD(p)S\tr)$ is a multi-affine
function of~$p$, that is, for each~$i$ we have $\psi(p)=c_1+c_2p_i$,
where $c_1$,~$c_2$ depend on~$i$ and~$p_j$ for $j\ne i$.

We claim that any multi-affine function~$\phi:[0,1]^k\to\R$ is non-zero on
the whole domain if and only if its values on the vertices~$\{0,1\}^k$
have all the same sign. Assuming this claim holds, we notice that
$\psi(0)=\det I_n=1>0$, so $\psi$~is non-zero on~$[0,1]^{n^2}$ if and
only if it is positive on~$\{0,1\}^k$.

To prove the claim, first suppose that $\phi$~is non-zero on~$[0,1]^k$
but there are two vertices $u,v\in\{0,1\}^k$ such that $\phi(u)<0$ and
$\phi(v)>0$. Following the path along the edges of~$\{0,1\}$, we will find
two vertices $u',v'\in\{0,1\}$ that differ in exactly one coordinate and
such that $\phi(u')<0$ and $\phi(v')>0$. Without loss of generality we
may assume that $u'_1=0$ and $v'_1=1$, while $u'_i=v'_i$ for $i\ge2$.
Let $x\in[0,1]^k$ be defined by $x_1=\phi(u')/(\phi(u')-\phi(v'))$
and $x_i=u'_i$ for $i\ge2$.  Then $\phi(x)=0$, a contradiction.

Conversely, if $\phi$~is positive (negative) on all the vertices, it
is easy to prove by induction on face dimension that $\phi$~is positive
(negative) in every internal point of each face.
\end{proof}

Lemma~\ref{lem:multaf} together with the discussion that precedes it
imply the following characterisation.

\begin{lemma}
\label{lem:rs}
Let $A$ be a non-singular matrix and let $R,S$ be defined as above.
Then the matrix interval $[A,A+\Delta]$ is singular if and only if
\[\det(I_n+A^{-1}RD(p)S\tr) \leq 0\]
for some $p\in\{0,1\}^{n^2}$.
\qed
\end{lemma}

In order to get $D(p)$ from the middle of the product to the beginning, we
use the following lemma, whose proof we present in the Appendix.

\begin{lemma}
\label{lem:Gan}
Let $F\in\R^{k\times n}$ and $G\in\R^{n\times k}$. Then $\det(I_k+FG)=\det(I_n+GF)$.
\qed
\end{lemma}

This fact can be exploited to prove the following equivalence.

\begin{theorem}
\label{thm:pmat}
Let $A$ be a non-singular matrix and let $R,S$ be defined as in Lemma~\ref{lem:rs}.
Then the matrix interval $[A,A+\Delta]$ is singular if and only if
the matrix $M=I_{n^2} +  S\tr\!A^{-1}R$ is not a P-matrix.
\end{theorem}

\begin{proof}
Because of Lemma~\ref{lem:Gan},
\[\psi(p)=\det(I_{n^2} + A^{-1}RD(p)S\tr)=\det(I_{n^2} + D(p)S\tr\!A^{-1}R).\]

If $p\in\{0,1\}^{n^2}$ and $p\ne0$, the expression
$\det(I_{n^2} + D(p)S\tr\!A^{-1}R)$ is equal to the principal minor
of the matrix~$M$ obtained by selecting exactly those rows and columns
that correspond to the $1$-entries of the vector~$p$.  Thus $\psi(p)$ is
non-positive for some $p\in\{0,1\}^{n^2}$ if and only if the matrix~$M$
is not a P-matrix.

The proof is now completed by applying Lemma~\ref{lem:rs}.
\end{proof}

\begin{cor}
The problem P-MATRIX is co-NP-complete.
\end{cor}

\begin{proof}
NP-hardness follows from Corollary~\ref{cor:sing} and Theorem~\ref{thm:pmat}.

The problem belongs to co-NP because after guessing the rows and columns,
the corresponding principal minor, which certifies the negative answer,
can be computed in polynomial time.
\end{proof}

\section*{Appendix: Proof of Lemma~\ref{lem:Gan}}

One of the basic facts about determinants is that adding a multiple
of a row to another row does not change the determinant. The
following lemma (Theorem~3 in Section~2.5 of Gantmacher's
book~\cite{Gan:The-Theory-of-Matrices}) is a block version of this
fact. Even though it holds for matrices with an arbitrary number of
blocks, we state it just for $2\times2$ blocks. This variant is sufficient
for the proof of Lemma~\ref{lem:Gan}.

\begin{lemma}
\label{lem:la}
Let $A\in\R^{m\times n}$ be a matrix with block structure
\[A=
\bordermatrix{
{\scriptstyle m_1}\{&\overbrace{A_{1,1}}^{n_1}&\overbrace{A_{1,2}}^{n_2}\cr
{\scriptstyle m_2}\{&A_{2,1}&A_{2,2}\cr
}
\]
and let $X\in\R^{m_1\times m_2}$, $Y\in\R^{n_1\times n_2}$. Then
\[
\det A=
\det\begin{pmatrix}
A_{1,1}+XA_{2,1}&A_{1,2}+XA_{2,2}\\
A_{2,1}&A_{2,2}
\end{pmatrix}\\
=
\det\begin{pmatrix}
A_{1,1}&A_{1,2}+A_{1,1}Y\\
A_{2,1}&A_{2,2}+A_{2,1}Y
\end{pmatrix}.
\]
\end{lemma}

\begin{proof}
Since
\[
\begin{pmatrix}
A_{1,1}+XA_{2,1}&A_{1,2}+XA_{2,2}\\
A_{2,1}&A_{2,2}
\end{pmatrix}
=
\begin{pmatrix}
I_{m_1}&X\\
0&I_{m_2}
\end{pmatrix}
A,
\]
we have
\[
\det\begin{pmatrix}
A_{1,1}+XA_{2,1}&A_{1,2}+XA_{2,2}\\
A_{2,1}&A_{2,2}
\end{pmatrix}
=
\det\begin{pmatrix}
I_{m_1}&X\\
0&I_{m_2}
\end{pmatrix}
\cdot
\det A
=
\det A.
\]
Similarly
\[
\det\begin{pmatrix}
A_{1,1}&A_{1,2}+A_{1,1}Y\\
A_{2,1}&A_{2,2}+A_{2,1}Y
\end{pmatrix}
=
\det A\cdot
\det
\begin{pmatrix}
I_{n_1}&Y\\
0&I_{n_2}
\end{pmatrix}
=\det A.
\qedhere
\]
\end{proof}

Finally comes the proof of Lemma~\ref{lem:Gan}.

\begin{proof}[Proof of Lemma~\ref{lem:Gan}]
Applying Lemma~\ref{lem:la} twice, we get
\begin{multline*}
\det(I_k+FG)=
\det \begin{pmatrix}
I_k+FG&0\\
G&I_n
\end{pmatrix}
\stackrel{(*)}{=}
\det\begin{pmatrix}
I_k&-F\\
G&I_n
\end{pmatrix}
\\
\stackrel{(\dagger)}{=}
\det\begin{pmatrix}
I_k&0\\
G&I_n+GF
\end{pmatrix}
=\det(I_n+GF).
\end{multline*}
Here $(*)$ follows by applying Lemma~\ref{lem:la} to rows with $X=F$
and $(\dagger)$ follows by applying it to columns with $Y=F$.
\end{proof}

\end{document}